\documentclass[11pt, letter, onecolumn]{IEEEtran}
\IEEEoverridecommandlockouts

\usepackage{pifont}
\usepackage[table]{xcolor}
\usepackage[T1]{fontenc}
\usepackage{tempora} 
\usepackage{cite,soul} 

\usepackage{mathtools,amscd,sarabian}
\usepackage[normalem]{ulem}
\usepackage{setspace}

\usepackage[inline, shortlabels]{enumitem}
\setlist{itemjoin ={,\enspace},itemjoin* = {, and\enspace}}
\usepackage{amsfonts,subfigure,lipsum,euscript}
\usepackage{amsthm}
\usepackage{amsmath,amsfonts}
\usepackage{amssymb}
\usepackage{graphicx}
\usepackage{mathrsfs}
\usepackage{dsfont}

\usepackage{url,rotating}
\usepackage{algorithm,algorithmic}
\usepackage[percent]{overpic}

\usepackage[font=footnotesize]{caption}
\usepackage{hyperref}
\hypersetup{
colorlinks=true,
linkcolor=black,
filecolor=black,      
urlcolor=black,
citecolor =black,
}

\newtheorem{lemma}{Lemma}    
\newtheorem{theorem}{Theorem}

\definecolor{skyblue}{RGB}{48,116,178}
\definecolor{darkred}{RGB}{222,61,62}
\definecolor{darkgreen}{RGB}{67,160,88}

\usepackage{inconsolata}

\usepackage{xspace}
\usepackage{dsfont}

\def\spc					 {SpecEst\xspace}
\def\uspc					{USF-SpecEst\xspace}

\def\adc					{ADC\xspace}

\def\usf					{USF\xspace}

\def\madc					{{$\mathscr{M}_\lambda$--{\fontsize{9pt}{9pt}\selectfont\texttt{ADC}}}\xspace}
\def\res          				{{\varepsilon}_{g}}
\def\taum                     {\tau_{m}}

\newcommand{\ceil}[1]	  {\left\lceil {#1} \right\rceil}
\newcommand{\Tl}[3]        		{\mathscr{T}_{#1}^{#2} \rob{#3}}

\newcommand\CRBL[1]                                 {\texttt{CRB}_{\lambda}\rob{#1}}
\newcommand\CRBW[1]                                {\texttt{CRB}_{\noise}\rob{#1}}

\def\ie					     {\emph{i.e.~}}

\newcommand{\MO}[1]		{\mathscr{M}_\lambda ({#1} )}

\newcommand{\mat}[1]		{\mathbf{#1}}

\newcommand\fig[1]			{Fig.~\ref{#1}}
\def\N					{\mathbb N}
\def\Z					{\mathbb Z}
\def\R					{\mathbb R}

\def\DE					{\stackrel{\rm{def}}{=}}
\newcommand{\id}[1]		{\mathbb{I}_{#1} }

\DeclareSymbolFont{cyrletters}{OT2}{wncyr}{m}{n}
\DeclareMathSymbol{\Sha}{\mathalpha}{cyrletters}{"58}

\newcommand\subfig[3]			{Fig.~\ref{#1} (${\mathsf{#2}}_{#3}$)}

\newcommand{\Lp}[1]{{\mathbf{L}}_{{#1}}}

\newcommand{\abs}[1]              	{\left| #1\right|}

\newcommand{\eqdef}{\stackrel{\rm def}{=}}

\def\ak         									{a_{k,n}}

\def\g 												  {g}

\def\mbg 												{\mat{\overline{g}}}

\newcommand{\EQc}[1]		{\stackrel{\eqref{#1}}{=}}

\def\momega                           {\boldsymbol{\omega}}

\def\crlb                                     {CRBs\xspace}

\def\by				                 	      {\overline{y}}
\def\byw				                 	{\overline{y}_{\noise}}
\def\mbyw				                 {\mat{\overline{y}}_{\noise}}
\def\bg				                 	      {\overline{g}}
\def\bres				                 	{\overline{\varepsilon}_{g}}
\def\bw				                 	     {x}
\def\nm				                 	     {n_m}
\def\ep                                       {\mathbb{E}}

\newcommand\Pro[1]				    {\mathsf{Pr}\rob{#1}}
\newcommand\cdf[2]				    {\mathsf{F}_{#1}\rob{#2}}
\newcommand\pdf[2]				    {\mathsf{p}_{#1}\rob{#2}}
\newcommand\secref[1]			{Section \ref{#1}}

\newcommand\thmref[1]			{Theorem \ref{#1}}

\newcommand\lemmaref[1]			{Lemma \ref{#1}}

\def\parameter       			{\underline{\Theta}}
\def\Log                      {{\mathcal{L}}\rob{\mbyw;\prvec}}
\def\logL                    {\log{\mathcal{L}}\rob{\mbyw;\prvec}}

\def\ZZ 							 {\mat{Z}}

\def\ZO                                  {\mathbf{Z}_1}
\def\ZT                                  {\mathbf{Z}_2}

\def\det               			{\operatorname{det}}

\newcommand{\bpara}[1]		{\smallskip \noindent {\bf #1}}

\newcommand{\BulPara}[1]	{\noindent {$\bullet$ {\bf #1}}}

\def\psnr					{\mathsf{PSNR}}

\def\v                          {\mat{\gbnoise}}
\def\mR                     {\mat{R}}

\newcommand{\transp}{{\top}}

\def\gradient          {{\partial_{\parameter}  {\mbg}}}

\def\ZO 				     {\mat{z}_{1}}
\def\ZT 				     {\mat{z}_{2}}
\def\ZTH 				   {\mat{z}_{3}}

\def\fun     				 {f}
\def\o               		   {\operatorname{o}}

\def\l						{\left(}
\def\r						{\right)}

\def\prvec					{\underline{\theta}}

\def\prvest				   {\hat{\underline{\theta}}}
\def\phik                    {\varphi_{k}}
\def\ak 					   {a_k}
\def\noise                 {w}
\def\gbnoise           {v}
\def\Gbnoise           {V}
\def\bNoise             {{X}}
\def\yw					    {y_{\noise}}
\def\sos 					{\texttt{SoS}_K\l \prvec\r}
\def\msep                 {\EuScript{E}_2{\rob{p_{\gbnoise}, p_{\bNoise}}}}
\def\IMO                    { \mathbf{I} \rob{\prvec}}

\newcommand\rob[1]			{\l #1 \r}
\newcommand{\sqb}[1]		{\left[ #1 \right]}
\newcommand{\ft}[1]			{\left[\kern-0.15em\left[#1\right]\kern-0.15em\right]}
\newcommand{\fe}[1]		{\left[\kern-0.30em\left[#1\right]\kern-0.30em\right]}
\newcommand{\flr}[1]		{\left\lfloor #1 \right\rfloor}
\makeatletter
\newcommand*{\rom}[1]{\expandafter\@slowromancap\romannumeral #1@}
\makeatother

\makeatletter
\def\moverlay{\mathpalette\mov@rlay}
\def\mov@rlay#1#2{\leavevmode\vtop{
\baselineskip\z@skip \lineskiplimit-\maxdimen
\ialign{\hfil$\m@th#1##$\hfil\cr#2\crcr}}}
\newcommand{\charfusion}[3][\mathord]{
#1{\ifx#1\mathop\vphantom{#2}\fi
\mathpalette\mov@rlay{#2\cr#3}
}
\ifx#1\mathop\expandafter\displaylimits\fi}
\makeatother

\usepackage{fancybox}
\newcommand{\PPbox}[1]{\vspace{4pt}\noindent {{\small \colorbox{black}{\color{white}{\sf #1}}}\hspace{4pt}}}
\usepackage{flushend}

\begin{document}

\title{
USF Spectral Estimation: \\  Prevalence of Gaussian Cram\'er-Rao Bounds \\ Despite Modulo Folding	
\thanks{This work is supported by the European Research Council's Starting Grant for ``CoSI-Fold'' (101166158) and UK Research and Innovation council's FLF Program ``Sensing Beyond Barriers via Non-Linearities'' (MRC Fellowship award no.~MR/Y003926/1). Further details on {Unlimited Sensing} and materials on \textit{reproducible research} are available via {\texttt{https://bit.ly/USF-Link}}.}} 

\author{ 
Ruiming Guo 
and
Ayush Bhandari \\ 
Dept. of Electrical and Electronic Engg., 
{Imperial College London}, SW7 2AZ, UK. \\ 
\texttt{\{ruiming.guo,a.bhandari\}@imperial.ac.uk}

\vspace{0.5cm}

{\color{blue} \sf To appear in Proc. of 2025 IEEE Statistical Signal Processing Workshop.}
}

\maketitle

\begin{abstract}

Spectral Estimation (SpecEst) is a core area of signal processing with a history spanning two centuries and applications across various fields. With the advent of digital acquisition, SpecEst algorithms have been widely applied to tasks like frequency super-resolution. However, conventional digital acquisition imposes a trade-off: for a fixed bit budget, one can optimize either signal dynamic range or digital resolution (noise floor), but not both simultaneously.
The Unlimited Sensing Framework (USF) overcomes this limitation using modulo non-linearity in analog hardware, enabling a novel approach to SpecEst (USF-SpecEst). However, USF-SpecEst requires new theoretical and algorithmic developments to handle folded samples effectively.
In this paper, we derive the Cramér-Rao Bounds (CRBs) for SpecEst with noisy modulo-folded samples and reveal a surprising result: the CRBs for USF-SpecEst are scaled versions of the Gaussian CRBs for conventional samples. Numerical experiments validate these bounds, providing a benchmark for USF-SpecEst and facilitating its practical deployment.

\end{abstract}


\tableofcontents

\newpage

\spacing{1.5}

\section{Introduction}
\label{sec:intro}

Spectral Estimation or \spc, also referred to as frequency estimation, is a classical and fundamental topic in the signal processing (SP) community, with numerous applications in radar, communications, acoustics, optics, and more. Since the pioneering work of Prony \cite{Prony:1795:J}, \spc has been extensively studied over the past decades \cite{Kay:1981:J,Hua:1990:J,Bhaskar:2013:J,Park:2018:J,Guo:2023:J}.

Despite significant theoretical advances \cite{Stoica:1989:J,Yao:1995:J,Rice:2001:J} and algorithmic developments \cite{Hua:1990:J,Bhaskar:2013:J,Park:2018:J}, practical applications of \spc remain constrained by limitations in digital acquisition. 
Analog-to-digital converters (ADCs) quantize sum-of-sinusoids (SoS) signals under a finite bit budget, posing a trade-off between signal dynamic range and digital resolution. This poses fundamental limitations for two core capabilities in \spc: resolving {\bf closely-spaced frequencies} and distinguishing between {\bf weak and strong targets}. 

The USF has recently emerged as a paradigm shift breaking the trade-off in conventional digitization; the USF enables simultaneous high dynamic range (HDR) and high digital resolution sensing \cite{Bhandari:2017:C,Bhandari:2020:Ja,Bhandari:2021:J,Shtendel:2023:J,Guo:2024:J}. The USF leverages modulo non-linearity in the analog domain, producing computationally encoded, folded samples. The resulting digital representation is of low dynamic range and is achieved through the mapping,
\begin{equation}
\label{eq:map}
\mathscr{M}_{\lambda}
:g \mapsto 2\lambda \left( {\fe{ {\frac{g + \lambda}{{2\lambda }}  } } - \frac{1}{2} } \right), \quad 
\ft{g} \DE g - \flr{g}
\end{equation}
where 
$\flr{g}  = \sup \left\{ {\left. {k \in \Z} \right|k \leqslant g} \right\}$ and $\lambda>0$ is the folding threshold. Hardware realization \cite{Bhandari:2021:J} of the \usf has shown that it overcomes fundamental bottlenecks in conventional digitization, delivering performance breakthroughs in the context of: 

\BulPara{Signal Dynamic Range.} 
\usf eliminates clipping or saturation which is widely reported problem while offering up to 60-fold dynamic range extension \cite{Zhu:2024:Ca} in practical setups. 

\BulPara{Digital Resolution.} 
Modulo folding in USF offers high digital resolution, up to 10 dB improvement in quantization noise in applications such as radar \cite{Feuillen:2023:C} and tomography \cite{Beckmann:2024:J}.

\bpara{Contributions.} The USF-based spectral estimation was first studied in \cite{Bhandari:2018:C}.
Akin to Prony's method, we have recently shown that $6K+4$ modulo samples suffice to recover $K$ sinusoids \cite{Guo:2024:J}. Here, Our goal is to investigate the performance bounds for single channel. 
While CRBs for \spc has been widely covered in literature \cite{Stoica:1989:J}, the context of \uspc is relatively new and unexplored, however, highly relevant given the advantages of the USF. This paper investigates the theoretical limits of \uspc from noisy, non-linearly folded samples. Concretely, we derive Cram\'er-Rao Bounds (CRBs) for \uspc. Our starting point is the justification of the noise model based on on-going hardware experiments \cite{Guo:2024:J} with modulo ADCs or \madc. There on, we derive the \crlb for both single and multiple sinusoidal estimation which is validated via numerical experiments. Somewhat surprisingly, our work shows that despite the presence of modulo non-linearity in the acquisition, \emph{the CRBs for \uspc match the Gaussian CRBs for \spc}.

Our work fundamentally differs from the recent work in \cite{Cheng:2024:J} in two ways: (i) we provide CRBs for \uspc problem, and, (ii) our noise model in \eqref{eq:samples}, different from \cite{Cheng:2024:J},  stems from hardware experiments and validation based on ongoing work on \uspc \cite{Zhu:2024:C,Guo:2024:J}.

\section{Problem Formulation}
\label{sec: introduction}

Let $\g\rob{t} \in\Lp{\infty}$ represent a sum-of-sinusoids (SoS),
$
g\rob{t}  =  \sum\nolimits_{k = 1}^{K}\ak \sin\rob{\omega_{k}t + \phik} \equiv \sos$
{where } $\prvec = \{\ak,  \omega_{k}, \phik\}_{k\in\id{K}}$ denotes the \emph{unknown} amplitudes, phases and frequencies, respectively. The action of folding non-linearity in \eqref{eq:map} maps $\g$ into a folded, continuous-time signal, $y\rob{t} = \MO{\g\rob{t}}$. Thereon, $y\rob{t}$ is pointwise sampled, leading to folded samples
$
y \sqb{n} = {\left.  \MO{g\rob{t}} \right|_{t = nT}}, n\in\id{N}
$
where $T$ is sampling step and $N$ is number of samples ($\id{N} = \{1,\ldots N\}$).
Since an \adc performs analog-domain folding, the noise contribution arising during sampling is attributed to (i) thermal noise following a Gaussian distribution \cite{Zhang:2015:J} and (ii) quantization noise following a uniform distribution. 
As a result, in real-world scenarios, the noisy measurements are modeled as,
\begin{equation}
\label{eq:samples}
\yw \sqb{n} = y \sqb{n} + \noise\sqb{n}, \quad n\in\id{N}
\end{equation}
where $\noise\sqb{n} \sim \mathcal{N}(0,\,\sigma^{2})$ denotes the noise. {We refer the reader to Fig. 6 in \cite{Guo:2024:J} for details on the hardware experiments that justify the noise assumption.} 
Given ${\yw\sqb{n}}_{n\in\id{N}}$, our goal is to derive statistical performance bounds, specifically the \crlb, to benchmark the performance of \usf-based spectral estimation.

\begin{figure*}[!t]	
\begin{align}
\label{eq:RN}
&\sqb{\mR}_{1,1} =  \gamma^{-1} N + \o(N), \;\;
{\abs{\sqb{\mR}_{1,2}}}  \leqslant \tfrac{\rob{2 + \abs{\cot(\omega_1 T)}}a_1 TN + \o(N)}{2} ,\;\;
{\sqb{\mR}_{1,3}} = \o(N) ,\;\;
\sqb{\mR}_{2,2} = \tfrac{ \gamma^{-1} a_1^2 T^2 N^3 }{3}
\notag \\
&  \sqb{\mR}_{2,3} = \tfrac{ \gamma^{-1} a_1^2T N^2 + \o(N^2)}{2}, \;\; 
\sqb{\mR}_{3,3}  = \gamma^{-1} a_1^2 N + \o(N) \quad \mbox{where} \quad \gamma \DE \rob{1 - \cos\rob{\omega_1 T}}^{-1}
\end{align}
\rule{\textwidth}{0.8pt}
\end{figure*}
\setcounter{equation}{5}

\section{Cram\'er-Rao Bounds}
\label{sec:CRB}

\bpara{CRB for Single Sinusoid with Oversampling.}
We begin our analysis with the simple case of a single sinusoid and examine how the sampling interval $T$ influences the theoretical limits of spectral estimation from noisy folded samples. Furthermore, we establish the connection to the conventional \crlb \cite{Stoica:1989:J, Rice:2001:J, Kay:1993:B}.
Our starting point is the modular decomposition property, 
\begin{align}
\label{eq:moddec}
g = \MO{g} + \res, 
\ \ 
\res\rob{t} = \sum\nolimits_{m=1}^{M} c_m u \rob{t-\taum}
\end{align}
where $\res$ is the residue characterized by $c_m \in 2\lambda\Z$ and $u\rob{\cdot}$ and $\{\taum\}_{m\in\id{M}}$ denote the unit step function and folding instants, respectively. 
Next, we demonstrate that oversampling inherently leads to bounded amplitudes $c_m$ of the residue.

\begin{lemma}
\label{lemma:1}
Let $ g(t)= \sos$, $\Delta^L = \Delta^{L-1}\circ\Delta$ denote the $L$-order finite-difference operator ($L\in\Z^{+}$) and $\kappa>0$ be an arbitrary constant. Then, 
\begin{equation}
\label{eq:os}
T \leqslant 
\frac{2}{\Omega }{\sin ^{ - 1}}\left( {\frac{1}{2}\sqrt[L]{{\frac{\kappa }{{{{\left\| {\mathbf{a}} \right\|}_{{\ell _1}}}}}}}} \right)
\Longrightarrow {\left\| \Delta {{\varepsilon _{{\Delta ^{(L-1)}}g}}} \right\|_\infty } \leqslant 2\lambda \ceil{\tfrac{\kappa}{2\lambda}}
\end{equation}
where $\mathbf{a}=[a_1,\ldots,a_K]^{\transp}$, ${\varepsilon}_s  \EQc{eq:moddec} s - \MO{s}$ and $\Omega \DE {\left\| \momega  \right\|}_\infty$.
\end{lemma}
\begin{proof}
By definition, the residue of ${\Delta^L}g$ is given by $  {{\varepsilon _{{\Delta ^L}g}}} =  \rob{ {{\Delta ^L}g} - \MO{{{\Delta ^L}g}}  } $. From modular arithmetic \cite{Bhandari:2020:Ja}, it follows that
$
{\left\| \Delta {{\varepsilon _{{\Delta ^{(L-1)}}g}}} \right\|_\infty } \leqslant 2\lambda \ceil{ {{\left\| {{\Delta ^L}g} \right\|_\infty }}/({2\lambda})}
$.
To establish the result, we first derive a bound on $\Delta^L g$. Consider the first-order difference $(\Delta g)\sqb{n} = g[n+1] - g\sqb{n}$. Using Young’s inequality for convolution, we have
${\left\| {\Delta g} \right\|_\infty } \leqslant  2\sin \left( {{{{\Omega }T}}/{2}} \right){\left\| {\mathbf{a}} \right\|_{{\ell _1}}}$. 
For an arbitrary $\kappa>0$, ${\left\| {\Delta g} \right\|_\infty }$ is upper-bounded by choosing
$
T \leqslant {{  { 2{{\sin }^{ - 1}}\left( {\frac{1}{2}{{\frac{\kappa }{{{{\left\| {\mathbf{a}} \right\|}_{{\ell _1}}}}}}}} \right)} }}/{{{\Omega }}}
$.
By induction, for $\Delta^L$, we obtain
\begin{equation}
\label{eq:bound}
{\left\| {{\Delta ^L}g} \right\|_\infty } \leqslant {2^L}{\sin ^L}\left( {{\Omega }T/2} \right){\left\| {\mathbf{a}} \right\|_{{\ell _1}}}.
\end{equation}
Thus, ${\left\| {{\Delta ^L}g} \right\|_\infty }$ can be bounded arbitrarily by choosing,
$
\notag
T \leqslant \Tl{\lambda}{\Omega}{\kappa} =
\rob{\tfrac{2}{\Omega}}{\sin ^{ - 1}}\left( {\tfrac{1}{2}\sqrt[L]{{\tfrac{\kappa }{{{{\left\| {\mathbf{a}} \right\|}_{{\ell _1}}}}}}}} \right)
\Longrightarrow {\left\| {{\Delta ^L}g} \right\|_\infty } \leqslant \kappa.
$
Substituting this bound into the inequality ${\left\| \Delta {{\varepsilon _{{\Delta ^{(L-1)}}g}}} \right\|_\infty } \leqslant 2\lambda \ceil{ {{\left\| {{\Delta ^L}g} \right\|_\infty }}/({2\lambda})}$ leads to the desired result. 
\end{proof}

\textcolor{black}{
\lemmaref{lemma:1} states that sampling at three times the Nyquist rate, \ie $T\leqslant \frac{\pi}{3\Omega}$, guarantees an exponential rate of dynamic range shrinkage. This constant-factor oversampling mitigates the resolution loss of the discrete frequencies $\{\omega_{k}T\}_{k=1}^{K}$ caused by oversampling. Moreover, choosing $\kappa\leqslant \lambda, L=1$ results in $M <N$ since 1) the discrete frequencies are upper bounded by $\frac{\pi}{3}$ and 2) ${\left\| {\Delta g} \right\|_\infty } \leqslant \lambda$ from \eqref{eq:bound}. }
In this paper, we focus on the case $L=1$. This choice is motivated by the fact that, in noisy scenarios, applying $\Delta^L$ with $L > 1$ amplifies noise, thereby degrading the accuracy of spectral estimation. The oversampling setup enables an asymptotic analysis of the \crlb in single sinusoid scenario, yielding explicit expressions for the \crlb as formalized below.
\begin{theorem}
\label{thm:2}
Let $g(t) = a_1 \sin\rob{\omega_1 t + \varphi_{1}}$ and the noisy folded samples be $\yw[n] = \MO{g(nT)} + \noise[n], n\in\id{N}$ with $\lambda = \abs{a_1} - \epsilon, \epsilon>0,
T \leqslant \left( {2/\abs{\omega _1}} \right){\sin ^{ - 1}}\left( {\lambda /\rob{2\abs{a_1}}} \right),
\{\omega_1 T, \varphi_{1}\} \in \mathbb{Q}^{2}$,
and $\noise[n]\sim \mathcal{N}(0,\,\sigma^{2})$. Then, 
\begin{align}
\label{eq:CRB}
&\mathop {\lim }\limits_{N \to \infty }  {\begin{bmatrix*}[l]
\CRBL{a_1} \\ 
\CRBL{\omega_1 T} \\ 
\CRBL{\varphi_{1}} 
\end{bmatrix*}}  
= {\frac{2 \gamma}{ N  \psnr}}
{\begin{bmatrix*}[c]
{a_1^2} \\ 
\rob{\sqrt{12}/N}^2 \\ 
\rob{2}^2
\end{bmatrix*}} \\
&\mbox{where} \quad \gamma = \rob{1 - \cos\rob{\omega_1 T}}^{-1}, \; \psnr = {a_1^2}/{\sigma^2}. \nonumber
\end{align}
Furthermore, $\lim\nolimits_{N\to\infty}\CRBL{\prvec} =  {\gamma}\lim\nolimits_{N\to\infty}\CRBW{\prvec}$
where $\CRBW{\cdot}$ denotes the conventional \crlb \cite{Stoica:1989:J}. 
\end{theorem}
\begin{proof}
To prove \eqref{eq:CRB}, we begin by modeling the residue as random impulsive noise governed by a Bernoulli distribution. The resulting hybrid Gaussian-Bernoulli noise is approximated by a single Gaussian distribution, and we show that the approximation error is bounded and vanishes as $N \to \infty$. Based on this noise characterization, we derive the \crlb and establish its correspondence with the conventional \crlb.

\PPbox{1}{\bf Statistical Modelling of Residue.} From the modulo decomposition property in \eqref{eq:moddec}, we obtain that $y[n] = g[n] - \res [n], n\in\id{N}$. Denote by $\by = y[n+1] - y[n]$, then
\begin{equation}
\by[n] = \bg[n] - \bres[n], \;\;  \bres[n] = \sum\nolimits_{m=1}^{M} c_m \delta \sqb{n-\nm}
\end{equation}
where $\nm = \frac{\taum}{T} \in \id{N}$ and $\delta\sqb{\cdot}$ is the discrete delta function. 
In view of \lemmaref{lemma:1}, $T \leqslant \left( {2/\abs{\omega _1}} \right){\sin ^{ - 1}}\left( {\lambda /\rob{2\abs{a_1}}} \right) \Longrightarrow c_{m} \in \{-2\lambda,2\lambda\}$.
Unlike the deterministic strategies used in residue recovery \cite{Bhandari:2021:J} and followed-up approaches,\cite{Guo:2023:C,Rudresh:2018:C,Azar:2022:C}, this paper models the residue $\bres\sqb{n}, n \in \id{N-1}$ as random impulsive noise following Bernoulli distribution, 
\begin{equation}
\label{eq:Bernoulli}
\Pro{\bres[n]=z} = p\delta[z - 2\lambda] + q \delta [z+2\lambda] + 1-(p+q) \delta [z]
\end{equation}
and $p+q= \frac{M}{N}$.
We assume the equal probability of positive and negative folds, resulting in $p=q=M/(2N)$. 

\PPbox{2}{\bf Noise PDF Approximation.}
Given this stochastic characterization, the noisy folded measurements $\byw$ read as
$\byw\sqb{n} = \bg\sqb{n} + \gbnoise\sqb{n}$, {where} $\gbnoise\sqb{n} = \bres\sqb{n} + \bw\sqb{n}$.
$\bw\sqb{n} = \noise\sqb{n+1} -\noise\sqb{n}$ satisfies that $\bw\sqb{n} \sim \mathcal{N}(0,\,2\sigma^{2})$ since $\{\noise\sqb{n}\}_{n\in\id{N}} \sim \mbox{IID}\; \mathcal{N}(0,\,\sigma^{2}) $ ({IID refers to independent and identically distribution}). 

We then analyze the PDF of $\gbnoise\sqb{n}, n\in\id{N-1}$. Denote by $\cdf{\Gbnoise}{\gbnoise}$ and $\pdf{\Gbnoise}{\gbnoise}$ the cumulative distribution function (CDF) and probability density function (PDF) of the random variable $\Gbnoise$, respectively. Then, by definition, we have
\begin{equation*}
\cdf{\Gbnoise}{\gbnoise} 
= \Pro{\Gbnoise \leqslant \gbnoise} 
= \ep\sqb{ \Pro{\bres \leqslant v - \bw} } = \rob{p_{\bw}*\mathsf{F}_{\bres}}\rob{\gbnoise}
\end{equation*}  
where $*$ denotes convolution and $\cdf{\bres}{\bres}$ is given by
\begin{equation*}
\cdf{\bres}{\bres} \EQc{eq:Bernoulli} q u \rob{\bres + 2\lambda } + (1 - p - q) u\rob{\bres} + p u\rob{\bres - 2\lambda}.
\end{equation*}
As a result, the CDF of $\gbnoise$ can be computed as
\begin{equation*}
\cdf{\Gbnoise}{\gbnoise}  = q\cdf{\bNoise}{\gbnoise + 2\lambda} + (1 - p - q)\cdf{\bNoise}{\gbnoise }  + p\cdf{\bNoise}{\gbnoise - 2\lambda} 
\end{equation*}
and hence, its PDF is given by
$
\pdf{\Gbnoise}{\gbnoise} = \frac{d \cdf{\Gbnoise}{\gbnoise}}{d\gbnoise} = q\pdf{\bNoise}{\gbnoise + 2\lambda} 
+ (1 - p - q)\pdf{\bNoise}{\gbnoise }  + p\pdf{\bNoise}{\gbnoise - 2\lambda} 
$
which is a mixture of Gaussians. Here, we approximate $\pdf{\Gbnoise}{\gbnoise}$ using $\pdf{\bNoise}{\gbnoise }$, with the approximation error upper bounded by
\begin{align}
\label{eq:MSE}
\msep &\eqdef \int_{\R} \abs{\pdf{\Gbnoise}{\gbnoise} - \pdf{\bNoise}{\gbnoise }}^{2} d\gbnoise \notag \\
&= \tfrac{p^{2}}{2\sqrt{2\pi}\sigma}\rob{  6 + 2e^{\tfrac{-4\lambda^{2}}{2\sigma^{2}}} - 8e^{\tfrac{-\lambda^{2}}{2\sigma^{2}}}}.
\end{align}
Note that $\msep$ is inversely proportional to $\sigma$ and $p^{2}$. 
This indicates that the approximation is accurate in scenarios with heavy noise or sparse foldings.

\PPbox{3}{\bf Asymptotic Analysis of Hybrid Noise.}
Given the hypotheses $\{\omega_1 T, \varphi_{1}\} \in \mathbb{Q}^{2}$, we can obtain that $\forall \{n,l\}\in\Z^{2}$, the equation $\omega_1 T n + \varphi_{1} = {\pi}/{2} + l\pi$ at most has one solution. This implies, given $N\in\Z^{+}, n\in\id{N}$, $\exists \epsilon > 0, \lambda = \abs{a_1} - \epsilon$, such that the inequality $\abs{g\rob{nT}}>\lambda$ at most has one solution. 

Combining these conditions, it follows that ${N \to \infty} \Longrightarrow M/N \to 0$. From \eqref{eq:Bernoulli}, we have $p = q = M/(2N)$, leading to the conclusion that ${N \to \infty} \Longrightarrow \msep \to 0$. This result indicates that the noise effect ($\gbnoise$) on spectral estimation from $\byw$ asymptotically reduces to that of $\bw$ alone, as if modulo folding does not influence the asymptotic behavior of the performance bounds. Therefore, in the following analysis, we approximate the PDF of $\gbnoise$ as $\pdf{\Gbnoise}{\gbnoise} = \pdf{\bNoise}{\gbnoise}$.

\PPbox{4}{\bf \crlb Derivations.} For an unbiased estimator $\prvest$ of $\prvec$
\begin{equation}
\ep \sqb{ ({ \prvest - \prvec }) ({ \prvest - \prvec })^{\transp}}
\succcurlyeq\left( \IMO \right)^{-1}
\end{equation}
where $\IMO$ is the \emph{Fisher Information Matrix} (FIM)
\begin{equation}
\label{eq:IM}
\IMO \eqdef \ep \left[{\partial_{\prvec} \logL} \rob{\partial_{\prvec} \logL}^{\transp} \right]
\end{equation}
and $\Log$ is the likelihood function of noisy samples ($\byw$). Due to the IID assumption, $\logL $ simplifies to,
\[
\logL 
=-(N-1)\log(2\sqrt{\pi} \sigma) - \tfrac{1}{4\sigma^{2}} \sum\nolimits_{n=1}^{N-1} \gbnoise^{2}[n].
\]
To compute $\IMO$, we need to evaluate ${\partial_{\prvec} \logL}$, 
$${\partial_{\prvec} \logL}= \tfrac{-1}{2\sigma^{2}} \partial_{\prvec}  (\mbyw - \mbg)^{\transp} 
\v= \tfrac{1}{2\sigma^{2}}  ({\partial_{\prvec}  {\mbg})^{\transp}} \v$$ 
where $\mbyw,\mbg $ and $\v$ denotes the vector form of $\byw[n], \bg[n]$ and $\gbnoise[n]$, respectively. Hence, $\IMO$ in \eqref{eq:IM} simplifies to,
\begin{equation}
\notag
\IMO =\tfrac{1}{4\sigma^{4}}\mathbb{E}\left[ {(\partial_{\prvec} \mbg)^{\transp}} 
{\v \v^{\transp}}  
{\partial_{\prvec} \mbg}   \right]  
= \tfrac{1}{2\sigma^{2}} 
\rob{{\partial_{\prvec} \mbg}}^{\transp}   {\partial_{\prvec} \mbg}.
\end{equation}
Given $\bg\sqb{n}  = g\rob{nT+T} - g\rob{nT} $, we obtain 
\begin{align*}
&{{\partial _{a_1} } \bg[n]}  = \sqb{ \ZO}_{n}= \mathsf{Im}(  e^{\jmath (\omega_1 T(n+1) + \varphi_{1})} - e^{\jmath (\omega_1 Tn + \varphi_{1})}) \\
&{{\partial _{\omega_1} } \bg[n]} =\sqb{ \ZT}_{n} = a_1 \mathsf{Re}(  e^{\jmath (\omega_1 T(n+1) + \varphi_{1}) } (nT+T) ) \\
&\qquad \qquad \qquad\qquad  - a_1 \mathsf{Re}(  e^{\jmath (\omega_1 Tn + \varphi_{1}) } nT )\\
&{{\partial _{\tau_1} } \bg[n]} = \sqb{ \ZTH}_{n} =  a_1 \mathsf{Re} ( e^{\jmath (\omega_1 T(n+1) + \varphi_{1}) } - e^{\jmath (\omega_1 Tn + \varphi_{1}) }   ).
\end{align*}
The Jacobian matrix $\partial_{\prvec} \mbg$ can be written in matrix form as
\begin{equation}
\label{eq:singradient}
\notag
\gradient 
= \ZZ  \quad \mbox{and} \quad
\ZZ = \begin{bmatrix}
\ZO\;\; \ZT\;\; \ZTH
\end{bmatrix} \in \R^{N\times 3}.
\end{equation}
Hence, 
$
\IMO = \frac{1}{2\sigma^{2}} \ZZ^{\transp} \ZZ
= \frac{1}{2\sigma^{2}} \mR
$
where $\sqb{\mR}_{i,j} = \mat{z}_{i}^{\transp} \mat{z}_{j}$. Next, we will characterize $\sqb{\mR}_{i,j}$ to elucidate the asymptotic bounds in explicit form. This requires the following computations for \emph{summations of trigonometric polynomials}. More specifically, 
with $\sum\nolimits_{n=1}^{N} e^{\jmath (n\omega+\varphi)} = \fun (\omega)$, we show that
\begin{equation}
\label{eq:lemma2}
\begin{bmatrix*}[c]
\sum\nolimits_{n=1}^{N} n^{L}\cos(n\omega+\varphi) \\ 
\sum\nolimits_{n=1}^{N} n^{L}\sin(n\omega+\varphi)
\end{bmatrix*} 
=  \begin{bmatrix*}[c]
\mathsf{Re} \rob{ (-\jmath)^{L} \partial_{\omega}^{L} \fun \rob{\omega}} \\
\mathsf{Im} \rob{ (-\jmath)^{L} \partial_{\omega}^{L} \fun \rob{\omega} } 
\end{bmatrix*}  
\end{equation}
where $\fun (\omega) \DE \tfrac{e^{\jmath \rob{ \omega+\varphi }} \rob{ 1 - e^{\jmath N\omega}}}{1 - e^{\jmath \omega}}$, $L\in \N$, $\{\omega,\varphi\} \in \R^{2}$.

We prove \eqref{eq:lemma2} by induction: for $L=0$, it boils down to geometric summation.  For $L>1$, taking the $L$-th order derivative of $\fun (\omega)$ leads to, 
$
\sum\nolimits_{n=1}^{N} (\jmath n)^{L} e^{\jmath (n\omega+\varphi)} = \partial_{\omega}^{L} \fun \rob{\omega}
$. From the Euler's formula, we obtain the desired result in \eqref{eq:lemma2} by separating the real and imaginary part of $\partial_{\omega}^{L} \fun \rob{\omega}$. 
Hence, substituting \eqref{eq:lemma2} into $\sqb{\mR}_{i,j}$ leads to the closed-form formula for $\sqb{\mR}_{i,j}$ in \eqref{eq:RN}, organized in terms of the order of $N$. $\o(\cdot)$ in \eqref{eq:RN} is the notion of order at infinity defined as, $b_n = \o(a_n) \Longleftrightarrow \lim\nolimits_{n\rightarrow \infty} {\abs{{a_n}/{b_n}}} = \infty$. With the result in \eqref{eq:RN}, the \crlb are computed by evaluating $(\IMO)^{-1} = {2\sigma^{2}} \mR^{-1}$, 
\begin{align}
\label{eq:Bound}
&{ {\begin{bmatrix*}[c]
\CRBL{a_1} \\			
\CRBL{\omega_1} \\ 
\CRBL{\tau_1}
\end{bmatrix*}} }{ \EQc{eq:RN} \tfrac{2\sigma^{2}}{\det(\mR)}  {\begin{bmatrix*}[l]
\gamma^{-2} a_1^4 T^2 N^4/12 + \o(N^4) \\
\gamma^{-2} a_1^2 N^{2} + \o(N^2)\\
\gamma^{-2} a_1^2 T^2 N^4/3 + \o(N^3)
\end{bmatrix*}} }. 
\end{align}
From \eqref{eq:RN}, we obtain $\det(\mR) = \gamma^{-3} a_1^4 T^2 N^5/12 + \o(N^5)$. By substituting $\det(\mR)$ into \eqref{eq:Bound}, we get the result in \eqref{eq:CRB}. 
\end{proof}
\bpara{Remarks.} The key insights from \thmref{thm:2} are as twofold: (i) The \crlb for sinusoidal parameter estimation from $\byw$ asymptotically converges to the conventional \crlb \cite{Stoica:1989:J}, scaled by a factor $\gamma$. (ii) The result holds \textcolor{black}{whenever $M = o(N)$. In this context, the sampling condition $\{\omega_1 T, \varphi_{1}\} \in \mathbb{Q}^{2}$ is just one possible scenario and is not a necessary condition.   }
\begin{figure}[!t]
\centering
\includegraphics[width=0.8\linewidth]{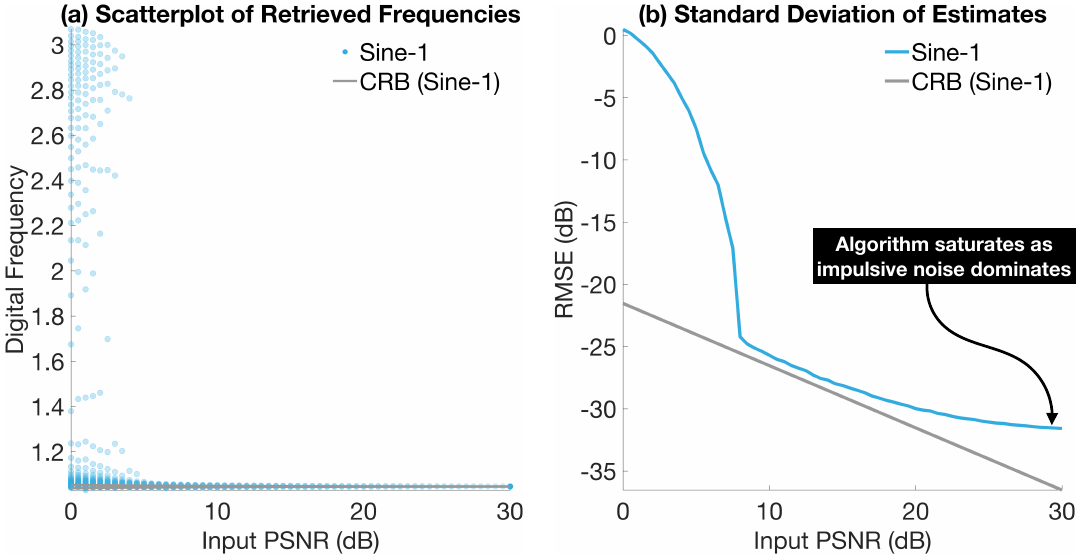}
\caption{
CRB tests for single sinusoid (averaged over $10000$ random realizations). 
We consider $K=1$ sinusoid with parameters: $\omega_1 T = 1.05$, $a_1 = 1$, $\epsilon = 9.83 \times 10^{-6}$, and $N = 100$ samples, resulting in $M = 6$ spikes.
(a) Scatterplot of the retrieved frequencies.
(b) Performance evaluation compared to the \crlb in \thmref{thm:2}.}
\label{fig:1}
\end{figure}

\bpara{CRB for Multiple Sinusoids with Oversampling.}
The analysis above reveals that the hybrid noise $\gbnoise$ asymptotically reduces to $\bw$ alone when $M = o(N)$. This result is independent of $g$, making it applicable to general multiple sinusoids scenarios. Consequently, the \crlb for multiple sinusoids from noisy folded samples $\byw$ converges to the conventional \crlb \cite{Kay:1981:J,Stoica:1989:J,Porat:1986:J}, scaled by a factor $\gamma_k = \rob{1 - \cos\rob{\omega_k T}}^{-1}$ for each frequency $\omega_k$. Notably, 
higher frequency moduli ($\abs{\omega_k}$) are expected to yield better estimation in the context of \usf.

Moreover, the convergence of $\gbnoise$ to $\bw$ suggests the feasibility of direct spectral estimation from $\byw$ in the modulo domain. This implies that existing high-resolution spectral estimation techniques can be adapted for this setting \cite{Kay:1981:J,Hua:1990:J,Guo:2022:J,Li:2021:J,Guo:2023:J}. Using a spectral estimation method, we empirically validate the result in \thmref{thm:2} through numerical experiments.

\section{Experiments}
\label{sec:exp}

Here we validate the \crlb derived in \thmref{thm:2}. Using the matrix pencil method \cite{Hua:1990:J}, we examine the gap between the theoretical limits (\thmref{thm:2}) and practical performance (algorithmic results) across various experimental settings. These include different noise levels, single and multiple sinusoids, and varying folding counts. Through numerical experiments, we demonstrate the validity of the derived \crlb and the corresponding analysis presented in \secref{sec:CRB}.

\begin{figure}[!t]
\centering
\includegraphics[width=0.8\linewidth]{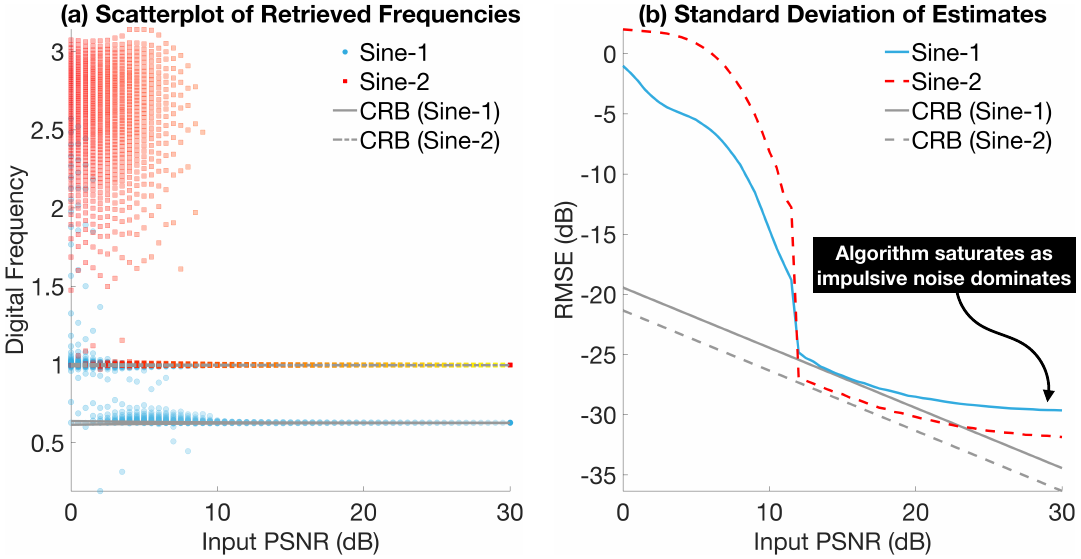}
\caption{
CRB tests for multiple sinusoids (averaged over $10000$ random realizations). 
We consider $K=2$ sinusoids with the following parameters: $\{\omega_1 T, \omega_2 T\} = \{0.63, 1.00\}$, $\{a_1, a_2\} =\{1, 1\}$, $\epsilon = 2.51 \times 10^{-4}$, and $N = 100$ samples, resulting in $M = 2$ spikes.
(a) Scatterplot of the retrieved frequencies.
(b) Performance evaluation. }
\label{fig:2}
\end{figure}

\bpara{CRB Tests for Single Sinusoid Case.} 
In the first experiment, we consider a single sinusoid with $\omega_1 T = 1.05$ and amplitude $a_1 = 1$. The threshold is set as $\lambda = \abs{a_1} - \epsilon$, with $\epsilon = 9.83 \times 10^{-6}$, resulting in $M = 6$ spikes from $N = 100$ samples.

The spectral estimation results obtained using the matrix pencil method \cite{Hua:1990:J} are presented in \fig{fig:1}. As shown in \subfig{fig:1}{b}{}, the algorithm's performance can be classified into three regions based on the PSNR value:

\PPbox{1}$\sqb{0,6}$ dB: 
In this region, the algorithm exhibits larger estimation variances due to significant noise.

\PPbox{2}$\sqb{6,16}$ dB: Beyond 6 dB, the algorithm's performance stabilizes and closely approaches the derived \crlb.

\PPbox{3}$\sqb{16,30}$ dB: 
In this low noise regime, $\msep$ becomes significant, meaning the impulsive noise $\bres$ dominates the noise effect ($\gbnoise$). Since $\bres$ is independent of PSNR, the algorithm's performance saturates and begins to deviate asymptotically from the derived \crlb.

\bpara{CRB Tests for Multiple Sinusoid Case.} 
We further evaluate the derived \crlb in the scenario of multiple sinusoids. In this experiment, the input consists of $K = 2$ sinusoids with $\{\omega_1 T, \omega_2 T\} = \{0.63, 1.00\}$ and amplitudes $\{a_1, a_2\} = \{1, 1\}$. The threshold is set as $\epsilon = 2.51 \times 10^{-4}$, resulting in $M = 2$ spikes from $N = 100$ samples.

As observed in the single sinusoid case, the algorithm's performance curve approaches the \crlb for $\psnr \in [11, 17]$ dB. Beyond this PSNR range, the algorithm saturates and reaches a performance ceiling due to an almost constant noise floor.
These experiments effectively validate the derived \crlb in \thmref{thm:2} and the analysis in \eqref{eq:MSE}, highlighting the potential performance ceiling for challenging scenarios such as resolving \textbf{closely-located} frequencies and separating \textbf{weak and strong} sinusoids.

\section{Conclusion}
\label{sec:conclusion}

Spectral estimation is a fundamental problem in signal processing, with the theoretical limits on parameter estimation in the presence of noise characterized by the Cram\'er-Rao bounds (\crlb). In this work, we studied spectral estimation in the context of the Unlimited Sensing Framework (\usf), where the signal is modulo-folded prior to sampling.
We derived the \crlb for spectral estimation from noisy modulo-folded measurements and established its relationship to the conventional \crlb. Numerical experiments were conducted to validate the derived bounds. These results provide a theoretical benchmark for spectral estimation in the \usf setting, offering insights into the performance limits of parameter estimation under non-linear sensing models.

\bibliographystyle{IEEEtran}

\end{document}